\documentclass[11pt]{article}

\usepackage[breaklinks]{hyperref}
\usepackage{graphicx}
\usepackage{theorem}
\usepackage{amssymb}
\usepackage{euscript}
\usepackage[T1]{fontenc}
\usepackage{color}
\usepackage{xspace}
\usepackage{amsmath}
\usepackage{enumerate}

\newtheorem{theorem}{Theorem}[section]
\newtheorem{observation}[theorem]{Observation}
\newtheorem{lemma}[theorem]{Lemma}
\newtheorem{claim}[theorem]{Claim}

{\theorembodyfont{\rm} \newtheorem{defn}[theorem]{Definition}}
   {\theorembodyfont{\rm} \newtheorem{remark}[theorem]{Remark}}

\newenvironment{proof}{{\em Proof:}}{\hfill{\hfill\rule{2mm}{2mm}}}

\newcommand{\MakeBig}{\rule[-.2cm]{0cm}{0.4cm}}
\newcommand{\sep}[1]{\,\left|\, {#1} \MakeBig\right.}
\newcommand{\cardin}[1]{\left| {#1} \right|}
\newcommand{\brc}[1]{\left\{ {#1} \right\}}
\newcommand{\pth}[2][\!]{#1\left({#2}\right)}
\newcommand{\pbrcx}[1]{\left[ {#1} \right]}
\newcommand{\Prob}[1]{\mathop{\mathbf{Pr}}\!\pbrcx{#1}}
\newcommand{\Ex}[1]{\mathop{\mathbf{E}}\!\pbrcx{#1}}
\newcommand{\Exs}[2]{\mathbf{E}_{#1}\!\pbrcx{#2}}
\newcommand{\ceil}[1]{\left\lceil {#1} \right\rceil}

\newcommand{\eps}{{\varepsilon}}
\renewcommand{\Re}{{\rm I\!\hspace{-0.025em} R}}
\providecommand{\Matousek}{Matou{\v s}ek\xspace}

\newcommand{\X}{\EuScript{M}}

\newcommand{\eqlab}[1]{\label{equation:#1}}
\newcommand{\Eqref}[1]{Eq.~(\ref{equation:#1})}

\newcommand{\lemlab}[1]{\label{lemma:#1}}
\newcommand{\lemref}[1]{Lemma~\ref{lemma:#1}}

\newcommand{\obslab}[1]{\label{observation:#1}}
\newcommand{\obsref}[1]{Observation~\ref{observation:#1}}

\newcommand{\defref}[1]{Definition~\ref{def:#1}}
\providecommand{\deflab}[1]{\label{def:#1}}

\newcommand{\atgen}{\symbol{'100}}
\newcommand{\SarielThanks}[1]{\thanks{Department of Computer
      Science;
      University of Illinois;
      201 N. Goodwin Avenue;
      Urbana, IL, 61801, USA;
      {\tt sariel\atgen{}uiuc.edu}; {\tt
         \url{http://www.uiuc.edu/\string~sariel/}.} #1}}
\providecommand{\TPDF}[2]{\texorpdfstring{#1}{#2}}
\definecolor{blue25}{rgb}{0,0,0.25}
\newcommand{\emphic}[2]{%
     \textcolor{blue25}{%
         \textbf{\emph{#1}}}%
         \index{#2}}
\newcommand{\emphi}[1]{\emphic{#1}{#1}}

\definecolor{shadecolor}{gray}{0.80}

\newcommand{\VBig}[1]{\rule[-.0cm]{-0.000cm}{#1}}

\DefineNamedColor{named}{OliveGreen}    {cmyk}{0.64,0,0.95,0.40}

\providecommand{\ComplexityClass}[1]{{{\textcolor[named]{OliveGreen}{
      \textsc{#1}}}}}
\providecommand{\NPHard}{{\ComplexityClass{NP-Hard}}\xspace}
\providecommand{\APX}{{\ComplexityClass{APX}}}
\newcommand{\Family}{\EuScript{F}}
\newcommand{\FamilyB}{\EuScript{G}}
\newcommand{\PntSet}{{\mathsf{P}}}
\newcommand{\PntSetB}{{\mathsf{Q}}}
\newcommand{\VC}{\ensuremath{\mathsf{VC}}\xspace}
\newcommand{\LP}{\ensuremath{\mathsf{LP}}\xspace}

\newcommand{\VCDim}{\delta}
\newcommand{\VCDimA}{\delta'}

\newcommand{\pnt}{\mathsf{p}}
\newcommand{\pntA}{\mathsf{q}}
\newcommand{\Opt}{{\mathrm{opt}}}

\newcommand{\PriceF}{\mathsf{f}}

\newcommand{\range}{r}
\newcommand{\rangeA}{{r}'}

\newcommand{\DemandChar}{\mathsf{d}}
\newcommand{\Demand}[1]{\mathsf{d}\pth{#1}}
\newcommand{\DemandExt}[2]{\mathsf{d}_{#1}\pth{#2}}

\newcommand{\REnergy}[1]{\mathsf{d}_{\mathrm{res}}\pth{#1}}

\newcommand{\RSample}{\EuScript{R}}
\newcommand{\Cover}[2]{\#\pth{#1 \cap #2}}
\newcommand{\Simplex}{\triangle}
\newcommand{\Instance}{\mathcal{I}}
\newcommand{\InstanceB}{\mathcal{J}}

\newcommand{\ArrX}[1]{\EuScript{A}\pth{#1}}
\newcommand{\ArrVD}[1]{\EuScript{A}_{\!||} \pth{#1}}
\newcommand{\Cell}{\triangle}
\newcommand{\CList}[1]{\mathrm{cl}\pth{#1}}
\newcommand{\Decomp}[1]{\mathrm{dcmp}\pth{#1}}
\renewcommand{\th}{t{}h\xspace}
\providecommand{\remove}[1]{}

\newcommand{\seclab}[1]{\label{sec:#1}}
\newcommand{\secref}[1]{Section~\ref{sec:#1}}

\newcommand{\apndlab}[1]{\label{apnd:#1}}
\newcommand{\apndref}[1]{Appendix~\ref{apnd:#1}}

\newcommand{\xbeginlgox}{\begin{minipage}{1in}\begin{tabbing}
        \quad\=\qquad\=\qquad\=\qquad\=\qquad\=\qquad\=\qquad\=\kill}
\newcommand{\xendlgox}{\end{tabbing}\end{minipage}}

\newcommand{\Cutting}{\EuScript{C}}

\newcommand{\FZero}{\mathcal{U}}
\newcommand{\FAtMost}[2]{f_{\leq {#1}}\pth{#2}}

\newcommand{\Bronniman}{Br{\"o}nnimann\xspace}

\newcommand{\thmref}[1]{Theorem~\ref{theo:#1}}
\newcommand{\thmlab}[1]{\label{theo:#1}}

\newcommand{\clmlab}[1]{\label{claim:#1}}
\newcommand{\clmref}[1]{Claim~\ref{claim:#1}}

\newcommand{\etal}{et al.\ }

\newcommand{\Universe}{\mathsf{U}}

\providecommand{\NP}{\ComplexityClass{NP}\xspace}

\providecommand{\POLYT}{\ComplexityClass{P}\xspace}
\providecommand{\DTIME}{\ComplexityClass{DTIME}}

\newcommand{\GroundSet}{\textsf{X}}
\newcommand{\RangeSpace}{\mathsf{S}}
\newcommand{\RangeSpaceA}{\mathsf{T}}
\newcommand{\RangeSet}{{\mathcal{R}}}
\newcommand{\RangeSetA}{\mathcal{R}'}
\newcommand{\RangeSetB}{\mathcal{Z}}

\IfFileExists{sariel_computer.sty}{\def\sarielComp{1}}{}
\ifx\sarielComp\undefined
   \newcommand{\SarielComp}[1]{}
   \newcommand{\NotSarielComp}[1]{#1}
\else
   \newcommand{\SarielComp}[1]{#1}%
   \newcommand{\NotSarielComp}[1]{}%
\fi

\begin{document}
\title{On the Set Multi-Cover Problem in Geometric Settings\thanks{A
      preliminary version of this paper appeared in Proc.\ of ACM SoCG, 2009
      \cite{cch-smcpg-09}.}}

\author{Chandra Chekuri%
   \thanks{Department of Computer Science, University of Illinois, 201
      N.\ Goodwin Ave., Urbana, IL 61801, USA. {\tt
         chekuri@cs.illinois.edu}.  Partially supported by NSF grants
      CCF-0728782 and CNS-0721899.  }%
   \and%
   Kenneth L. Clarkson%
   \thanks{IBM Almaden Research Center, San Jose, CA 95120, USA. {\tt
         kclarks@us.ibm.cam}.}  \and %
   Sariel Har-Peled\SarielThanks{}}

\date{\today}
\maketitle

\begin{abstract}
    We consider the set multi-cover problem in geometric
    settings. Given a set of points $\PntSet$ and a collection of
    geometric shapes (or sets) $\Family$, we wish to find a minimum
    cardinality subset of $\Family$ such that each point $\pnt \in
    \PntSet$ is covered by (contained in) at least $\Demand{\pnt}$
    sets. Here $\Demand{\pnt}$ is an integer demand (requirement) for
    $\pnt$. When the demands $\Demand{\pnt}=1$ for all $\pnt$, this is
    the standard set cover problem. The set cover problem in geometric
    settings admits an approximation ratio that is better than that
    for the general version. In this paper, we show that similar
    improvements can be obtained for the multi-cover problem as well.
    In particular, we obtain an $O(\log \Opt)$ approximation for set
    systems of bounded \VC-dimension, where $\Opt$ is the cardinality
    of an optimal solution, and an $O(1)$ approximation for
    covering points by half-spaces in three dimensions and for some
    other classes of shapes.
\end{abstract}



\section{Introduction}
\seclab{introduction}

The \emphi{set cover} problem is the following. Given a universe
$\Universe$ of $n$ elements and a collection of sets $\Family =
\brc{S_1, \ldots, S_m}$ where each $S_i$ is a subset of $\Universe$,
find a minimum cardinality sub-collection $C \subseteq \Family$ such
that $C$ covers $\Universe$; in other words, the union of the sets in
$C$ is $\Universe$. In the weighted version each set $S_i$ has a
non-negative weight $w_i$ and the goal is to find a minimum weight
cover $C$. In this paper, we are primarily interested in a
generalization of the set cover problem, namely, the \emphi{set
   multi-cover} problem. In this version, each element $e \in
\Universe$ has an integer demand or requirement $\Demand{e}$ and a
multi-cover is a sub-collection $C \subseteq \Family$ such that for
each $e \in \Universe$ there are $\Demand{e}$ \emph{distinct} sets in
$C$ that contain $e$.%
\footnote{A related and somewhat easier variant allows a set to be
   picked multiple times.  In this paper, unless explicitly stated, we
   use ``multi-cover'' for the version where only one copy of a set is
   allowed to be picked.} %
The set cover problem and its variants arise directly and indirectly
in a wide variety of settings and have numerous applications. Often
$\Family$ is available only implicitly, and could have size $m$
exponential in the size of $\Universe$, or even infinite (for example
$\Family$ could be the set of all disks in the plane). The set cover
problem is \NPHard and consequently approximation algorithms for it
have received considerable attention. A simple greedy algorithm, that
iteratively adds a set from $\Family$ that covers the most uncovered
elements, is known to give a $(1 + \ln n)$ approximation, where $n =
\cardin{\Universe}$. (In the weighted case, the algorithm picks the
set with minimum average cost for the uncovered elements.)  Similar
bounds can also be achieved via rounding a linear programming
relaxation.  The advantage of the greedy algorithm is that it is also
applicable in settings where $\Family$ is given implicitly, but there
exists a polynomial time oracle to (approximately) implement the
greedy step in each iteration. It is also known that unless
$\POLYT=\NP$ there is no $o(\log n)$ approximation for the set cover
problem \cite{ly-hamp-94}. Moreover, unless $\NP \subset
\DTIME(n^{O(\log \log n)})$ there is no $(1-o(1)) \ln n$ approximation
\cite{f-tlasc-98}. Thus the approximability of the general set cover
problem is essentially resolved if $\POLYT \neq \NP$. However, there
are many set systems of interest for which the hardness of
approximation result does not apply.  There has been considerable
effort to understand the approximability of set cover in restricted
settings, and previous work has shown that the set cover problem
admits improved approximation ratios in various geometric cases.  In
particular, set systems that arise in geometric settings are the focus
of this paper.

In the geometric setting, we use $(\PntSet, \Family)$ to describe a
set system (also referred to as a range space) where $\PntSet$ is a
set of points and $\Family$ is a collection of sets (also called
\emph{objects} or \emph{ranges}). We are typically interested in the
case where $\Family$ is a set of ``well-behaved shapes''. Examples of
such shapes include disks, pseudo-disks, and convex polygons. The goal
is to cover a given finite set of points $\PntSet$ in $\Re^d$ by a
collection of objects from $\Family$. At a higher level of
abstraction, one can consider set systems of small (or constant) \VC
dimension. In addition to the inherent theoretical interest in
geometric set systems, there is also motivation from applications in
wireless and sensor networks. In these applications the coverage of a
wireless or sensor transmitter can be reasonably approximated as a
disk-like region in the plane. The problem of locating transmitters to
optimize various metrics of coverage and connectivity is a
well-studied topic; see \cite{twdj-cpwsn-08} for a survey.

\Bronniman and Goodrich \cite{bg-aoscf-95}, extending the work of
Clarkson \cite{c-apca-93}, used the reweighting technique to give an
$O(\log \Opt)$ approximation for the set cover problem when the \VC
dimension of the set system is bounded%
\footnote{\Bronniman and Goodrich \cite{bg-aoscf-95} consider the {\em
      hitting set} problem which is the set cover problem in the dual
   range space. In this paper we blur the distinction between set
   cover and hitting set.}%
.  Here $\Opt$ is the size of an optimum solution. Known hardness
results \cite{ly-hamp-94} preclude such an approximation ratio for the
general set cover problem. The reweighting technique and its
application to set cover \cite{c-apca-93, bg-aoscf-95} show that the
approximation ratio for set cover can be related to bounds on
$\eps$-nets for set systems.  Using this observation,
\cite{bg-aoscf-95} showed an improved $O(1)$ approximation ratio for
the set cover problem in some cases, including the problem of covering
points by disks in the plane. Long \cite{l-upsda-01} made an explicit
connection between the integrality gap of the natural \LP relaxation
for the set cover problem and bounds on the $\eps$-nets for the set
system (see also \cite{ers-hsvcs-05}).  This allows $\Opt$ in the
approximation ratio to be replaced by $\PriceF$, where $\PriceF$ is
the value of an optimum solution to the \LP relaxation (i.e., the
optimal fractional solution). Clarkson and Varadarajan
\cite{cv-iaags-07} developed a framework to obtain useful bounds on
the $\eps$-net size via bounds on the union complexity of a set of
geometric shapes. Using this framework they gave improved
approximations for various set systems/shapes.  Recently, Aronov, Ezra
and Sharir \cite{aes-ssena-09}, and Varadarajan \cite{v-enuc-09}
sharpen the bounds of Clarkson and Varadarajan in some cases
\cite{cv-iaags-07}.

The geometric set cover problem induced by covering points by disks in
the plane is strongly \NPHard \cite{fg-oafac-88}; very recently a PTAS
was obtained for this problem \cite{mr-pghsp-09} improving a
previously known constant factor approximation.  Some other geometric
coverage problems are known to be \APX-hard \cite{fmz-mgbag-07}; that
is, there is a constant $c > 1$ such that unless $\POLYT=\NP$, there
is no $c$ approximation for them.

\paragraph{Our results.}
In this paper, we consider the multi-cover problem in the geometric
setting. In addition to the set system $(\PntSet,\Family)$, each point
$\pnt \in \PntSet$ has an integer demand $\Demand{\pnt}$. Now a 
cover needs to include, for each point $\pnt$,
$\Demand{\pnt}$ sets that contain $\pnt$.  For
general set systems, the greedy algorithm and other methods such as
randomized rounding, which work for the set cover problem, can be
adapted to the multi-cover problem, yielding a $(1 + \ln n)$
approximation (see \cite{v-aa-01}). In contrast, the $\eps$-net based
approach for geometric set cover does not generalize to the
multi-cover setting in a straight-forward fashion. Nevertheless, we
are able to use related ideas, in a somewhat more sophisticated way,
to obtain approximation ratios for the geometric set multi-cover
problem that essentially match the ratios known for the corresponding
set cover problem. In particular, we obtain the following bounds. In
all the bounds, $\PriceF \le \Opt$ is the value of an optimum
(fractional) solution to the natural \LP relaxation, and $\Opt$ is the
value of an optimum (integral) solution.
\begin{itemize}
    \item $O(\log \PriceF)$ approximation for set
    multi-cover of set systems of bounded \VC dimension.

    \item $O(1)$ approximation for (multi) covering points in $\Re^3$
    by halfspaces. This immediately leads to a similar result for
    multi-cover of disks by points in the plane.

    \item $O(\log \log \log \PriceF)$ approximation for covering
    points by fat triangles (or other fat convex polygonal shapes of
    constant descriptive complexity) in the plane.
\end{itemize}

The second and third results follow from a general framework for a
class of ``well-behaved'' shapes based on the union complexity of the
shapes. This is inspired by a similar framework from
\cite{cv-iaags-07, aes-ssena-09}. Our work differs from previous work
for set cover in geometric settings in two ways. First, we use the \LP
relaxation in an explicit fashion in several ways, demonstrating its
effectiveness. Second, our work points out the usefulness of shallow
cuttings for the multi-cover problem. We hope that these directions
will be further developed in the future.


\section{Preliminaries}
\seclab{preliminaries}



\subsection{Problem statement and notation}

Let $\Instance=(\PntSet, \Family)$ be a given set system with \VC
dimension $\VCDim$. Here $\PntSet$ is a set of points, and $\Family$
is a collection of subsets of $\PntSet$, called \emph{ranges}
or \emph{objects}.  Assume that every point $\pnt \in
\PntSet$ has an associated integral \emphi{demand} $\DemandExt{\Instance}{\pnt}
\geq 0$. When the relevant set system is understood, we may write
$\Demand{\pnt}$.  Here we would like to find a minimum cardinality set of
ranges of $\Family$ that covers $\PntSet$, such that every $\pnt \in
\PntSet$ is covered at least $\Demand{\pnt}$ times.
Note that we allow a range of $\Family$ to be
included only once in the cover. This is an instance of the \emphi{set
multi-cover} problem. There is also a weaker version of the problem, where
the solution may be a multiset; that is, a range may be included multiple times.

We will also discuss the demand of a set $\PntSet'\subset\PntSet$,
which is $\Demand{\PntSet'} = \DemandExt{\Instance}{\PntSet'}=
\sum_{\pnt\in\PntSet'} \Demand{\pnt}$.  The \emphi{total demand} of a
set system $\Instance=(\PntSet, \Family)$ is $\Demand{\PntSet}$.


\begin{defn}
  For a point $\pnt \in \PntSet$ and a set $X \subseteq \Family$ where
  each range in $\Family$ has a non-negative weight, let $\Cover{\pnt}{X}$
  denote the \emphi{depth} of $\pnt$ in $X$; namely, it is the total
  weight of the ranges of $X$ covering $\pnt$. If the ranges do not
  have weights then we treat them as having weight one.
  \deflab{depth}
\end{defn}

\begin{defn}
  Given a multiset $\RangeSetB \subseteq \Family$, let $\InstanceB =
  (\PntSetB, \FamilyB) = (\PntSet, \Family) \setminus \RangeSetB$
  denote the \emphi{residual} set system.  The residual instance
  encodes what remains to be covered after we use the coverage
  provided by $\RangeSetB$. Each $\pnt\in\PntSet$ has \emphi{residual
    demand} $\REnergy{\pnt, \RangeSetB} = \max \pth{\Demand{\pnt} -
    \Cover{\pnt}{\RangeSetB}, 0}$, and $\PntSetB$ comprises the points
  of $\PntSet$ with nonzero residual demand.  Thus
  $\DemandExt{\InstanceB}{\pnt} = \REnergy{\pnt, \RangeSetB}$.  Also
  $\FamilyB = \Family \setminus \RangeSetB$.  We will also write, for
  $\PntSetB'\subset\PntSetB$, $\REnergy{\PntSetB', \RangeSetB} =
  \sum_{\pnt\in\PntSetB'} \REnergy{\pnt, \RangeSetB}$.  In particular,
  $\REnergy{\PntSetB, \RangeSetB} = \DemandExt{\InstanceB}{\PntSetB}$
  is the \emphi{total residual demand} of $\Instance$, with respect to
  $\PntSetB$.

    \deflab{residual}
\end{defn}

A set system $(\PntSet, \Family)$ has \VC dimension $\VCDim$ if no
subset of $\PntSet$ of cardinality greater than $\VCDim$ is
\emphi{shattered} by $\Family$. Here a set $\PntSet' \subseteq
\PntSet$ is shattered if for every $X \subset \PntSet'$ there is a
range $\range \in \Family$ such that $X = \range \cap \PntSet'$. Given
a range space $S = (\PntSet, \Family)$, its dual set system is $S^* =
(\Family, \PntSet^*)$ where $\PntSet^* = \{ \Family_{\pnt} \mid \pnt
\in \PntSet\}$ and $\Family_\pnt = \{ \range \in \Family \mid \pnt \in
\range\}$.  For a set system $S$ with \VC dimension $\VCDim$, we
denote by $\VCDim^*$ the \VC dimension of $S^*$. It is known that
$\VCDim^* \le 2^{\VCDim+1}$ \cite{pa-cg-95, h-gaa-08}; thus if $S$ has
bounded \VC dimension, so does $S^*$. However, for specific set
systems of interest, in particular geometric set systems, one can
directly show much stronger upper bounds on $\VCDim^*$.

\subsection{\TPDF{\LP}{L{}P} relaxation}

A standard approach to computing an approximate solution to an NP-hard
problem is to solve a linear programming relaxation (\LP) of the
problem and round its fractional solution to an integral solution to
the original problem.

In our case, if $\Family = \brc{\range_1, \ldots, \range_m}$ and
$\PntSet = \brc{\pnt_1, \ldots, \pnt_n}$, the natural \LP has a
variable $x_i$ for range $\range_i$:
\begin{align}
    \min\;\;\; &\sum_{i=1}^m x_i  & \nonumber \\
    \text{subject to}\;\;\;& \sum_{i: \pnt_j \in \range_i} x_i  \geq \Demand{\pnt_j}\;\;\;
    &\forall \pnt_j \in \PntSet,
    \eqlab{l:p}\\
    &x_i  \in [0,1] & i=1,\ldots, m.\nonumber
\end{align}

Note that \LP is a relaxation of the integer program for the set
multi-cover problem, for which $x_i$ are required to take a value in
$\brc{0,1}$. If repetitions of a set are allowed, then the constraint
$x_i \in [0,1]$ is replaced by $x_i \ge 0$.

Let $\PriceF = \PriceF(\Instance)$ denote the \emphi{value} of an
optimum solution to the above \LP. Clearly, $\Opt \ge
\PriceF(\Instance)$. We will refer to the values assigned to the
variables $x_i$ for some particular optimal solution to the \LP as the
\emphi{fractional solution}.  In the following, we will refer to the
value of $x_i$ in the solution as the \emphi{weight} of the range
$\range_i$. We will sometimes use vectors that are not optimal
solutions for \LP, but only \emphi{feasible}; that is, they
satisfy the constraints.

\subsection{Overview of Rounding for Geometric Set Cover}
\label{sec:overview}
We briefly explain the previous approaches for obtaining approximation
algorithms for the set cover problem in geometric settings.  The work
of Clarkson \cite{c-apca-93} and \Bronniman and Goodrich
\cite{bg-aoscf-95} used the reweighting technique and $\eps$-nets to
obtain algorithms that provide approximation bounds with respect to
the integer optimum solution. In \cite{l-upsda-01, ers-hsvcs-05}, it
was pointed out that these results can be reinterpreted as rounding
the \LP relaxation and hence the approximation bounds can also be
stated with respect to the fractional optimum solution. Here we
discuss this interpretation.

Note that in the set cover setting $\Demand{\pnt} = 1$ for all points.
Consider a fractional solution to the \LP given by $x_i$ assigned to
ranges $\range_i \in \Family$, with total value $\PriceF = \sum_i
x_i$. Let $\eps = 1/\PriceF$.  From the constraint (\Eqref{l:p}) it
follows that for each $\pnt$, $\sum_{i: \pnt \in \range_i} x_i/\PriceF
\ge \Demand{\pnt}/\PriceF = 1/\PriceF = \eps$. Interpreting
$x_i/\PriceF$ as the weight of range $\range_i$, we obtain a set
system in which all points are covered to within a weight of
$\eps$. Therefore an $\eps$-net of the (weighted) {\em dual} range
space is a set cover for the original instance. Now one can plug known
results on the size of $\eps$-nets for set systems to immediately
derive an approximation. For example, set systems with \VC dimension
$\VCDim$ have $\eps$-nets of size $O(\VCDim/\eps \cdot \log 1/\eps)$
\cite{pa-cg-95} and hence one concludes that there is a set cover of
size $O(\VCDim^* \PriceF \log \PriceF )$ computable in polynomial
time, that is, an $O(\VCDim^* \log \PriceF)$ approximation.  For some
set systems improved bounds on the $\eps$-net size are known. For
example, if $\PntSet$ is a finite set of points and $\Family$ is a set
of disks in the plane then $\eps$-nets of size $O(1/\eps)$ are known
to exist for the dual set system and hence one obtains an $O(1)$
approximation for covering points by disks in the plane.  Clarkson and
Varadarajan \cite{cv-iaags-07} showed that bounds on the size of
$\eps$-nets can be obtained in the geometric setting from bounds on
the union complexity of objects in $\Family$. We remark that the connection
to $\eps$-nets above also holds in the converse direction: for a given
set system, the integrality gap of \LP can be used to obtain bounds on the
$\eps$-net size.

In the multi-cover setting we can take the same approach as above.
However, now we have for a point $\pnt$, $\sum_{i: \pnt \in \range_i}
x_i/\PriceF \ge \Demand{\pnt} \cdot \eps$ where $\eps = 1/\PriceF$.
Note that we now have non-uniformity due to different demands and
hence an $\eps$-net would not yield a feasible multi-cover for the
original problem.


\section{Multi-cover in spaces with bounded \TPDF{\VC}{VC} dimension}

In this section, we prove the following theorem.

\begin{theorem}
    Let $\Instance=(\PntSet, \Family)$ be an instance of multi-cover
    with \VC dimension $\VCDim$ and let $\VCDim^*$ be the \VC
    dimension of the dual set system. There is a randomized poly-time
    algorithm that on input $\Instance$ outputs $O(\VCDim^* \PriceF
    \log \PriceF)$ sets of $\Family$ that together satisfy
    $\Instance$, where $\PriceF$ is the value of an optimum fractional
    solution to $\Instance$.  

    \thmlab{v:c:dim}
\end{theorem}

We have an easy proof of the above theorem for the setting in which a
set is allowed to be used multiple times; the proof is based on
results on relative approximations.  See \secref{with:reps} for
details.

It may be possible to adapt this proof to prove the above theorem for
the setting in which a set is not allowed to be included more than
once.  This, however, appears to be nontrivial and instead we next give a
proof, in \secref{no:repetition}, that uses the \LP to reduce the
problem to a regular set cover problem with a modified set system
whose primal and dual \VC dimensions are at most $O(\VCDim)$ and
$O\pth{\VCDim^*}$, respectively.

\subsection{Multi-cover without repetition}
\seclab{no:repetition}

\paragraph{Geometric intuition.}
Imagine we have a set of points and a set of disks
$\Family=\brc{\range_1, \ldots, \range_m}$ (i.e., the ranges) in the
plane. We solve the \LP for this system. This results in weight
assigned to each disk, such that the total weight of the disks
covering a point $\pnt \in \PntSet$ exceeds (or meets) its demand
$\Demand{\pnt}$. We add another dimension (we are now in three
dimensions), and for each $i=1,\ldots,m$ translate the disk $\range_i
\in \Family$ to the plane $z=i$. Let $\Family'$ denote the resulting
set of $m$ two-dimensional disks that ``live'' in three
dimensions. Observe that the projection of $\Family'$ to the $xy$
plane is $\Family$. Every point $\pnt_j \in \PntSet$ is now a vertical
line $\ell_j$ (parallel to the $z$-axis), and we are asking for a
subset $X$ of $\Family'$, such that every line $\ell_j$ stabs at least
$\Demand{\pnt_j}$ disks of $X$. The fractional solution for the
original problem induces a fractional solution to the new problem.
The next step, is to break every line $\ell_j$ into segments, such
that the total weight of the disks of $\Family'$ intersecting a
vertical segment is at least $1$ (and at most $2$).  Let $L'$ be this
resulting set of segments. Consider the ``set system'' $\RangeSpace =
(L', \Family')$, and its associated set cover instance of the disks of
$\Family'$ so that they intersect all the segments of $L'$. It is easy
to verify that any solution of this set cover problem, is in fact a
solution to the original multi-cover problem, and vice versa (up to
small constant multiplicative error, say 2). We know how to solve such
set-cover problems using standard tools. The key observation is that
the projection of $(L', \Family')$ on to the plane yields the original
range space. Similarly, projecting $(L', \Family')$ on to the $z$-axis
results in a range space where the points are on the real line and the
ranges are intervals. Since the range space $(L', \Family')$ is the
intersection of two range spaces of low \VC dimension, it has low \VC
dimension. This implies that the set-cover problem on $(L', \Family')$
has a good approximation \cite{bg-aoscf-95} and this leads to a good
approximation to the original multi-cover problem on $\RangeSpace$.

\paragraph{More formal solution.}
Consider a fractional solution $x$ to the \LP associated with
$\Instance$. If any set $\range_i \in \Family$ satisfies $x_i \ge 1/4$
then we add $\range_i$ to our solution.  There can be at most $4\sum_i
x_i = 4\PriceF$ such sets, so including them does not harm our goal of
a solution with $O(\PriceF)$ sets.  We now work with the residual
instance and hence we can assume that the fractional solution has no
set $\range_i$ with $x_i \ge 1/4$.

Now, assume that we have fixed the numbering of the ranges of $\Family =
\brc{ \range_1, \ldots, \range_m }$, and consider the fractional
solution, with the value $x_i$ associated with $\range_i$, see
\Eqref{l:p}. In particular, for a point $\pnt \in \PntSet$, consider
the linear inequality
\[
\sum_{i: \pnt \in \range_i} x_i \geq \Demand{\pnt}.
\]
This inequality holds for the fractional solution. We split this
inequality into $O(\Demand{\pnt})$ inequalities having $1/2$ on the
right hand side. To this end, scan this inequality from left to right,
and collect enough terms on the left-hand side, such that their sum
(in the fractional solution) is larger than $1/2$. We will write down
the resulting inequality, and continue in this fashion until all the
terms of this inequality are exhausted.

Formally, let $U_0 = U = \brc{i \sep{ \pnt \in \range_i}}$ be the
sequence of indices of the ranges participating in the above
summation, where $U$ and $U_0$ are sorted in increasing order.  For
$\ell \geq 1$, let $V_\ell$ be the shortest prefix of $U_{\ell-1}$ such that
$\sum_{i \in V_\ell} x_i \geq 1/2$, and let $u_\ell$ be the largest
number (i.e., index) in $V_\ell$, and let $U_\ell = \pth{U_{\ell-1}
  \setminus V_\ell}$. Since each $x_i < 1/4$ we have that $\sum_{i \in
  V_\ell} x_i < 1/2+1/4 < 3/4$. We stop when $\sum_{i \in U_\ell} x_i
< 1/2$ for the first time. This process creates some $h$
inequalities of the form
\[
\sum_{i \in V_\ell} x_i \geq 1/2,
\]
for $\ell = 1, \ldots, h$. We have $h\ge \Demand{\pnt}$ inequalities
from the fact that $\sum_{i: \pnt \in \range_i} x_i \ge \Demand{\pnt}$
and by our observation that $\sum_{i \in V_\ell} x_i < 3/4$.

\newcommand{\FA}[1]{\widetilde{#1}}
\newcommand{\FB}[1]{\overline{#1}}
\newcommand{\FC}[1]{\widehat{#1}}

We next describe a new set system $(\PntSet', \FC{\Family})$,
derived from this construction of inequalities, such that
a set cover solution to the new system implies a multi-cover
solution to the original system, and the new system
has small \VC dimension.

The new set system $(\PntSet', \FC{\Family})$ is defined as
follows. For each point $\pnt$ which was processed as above, we create $h$
copies of it, one for each $V_\ell$. Each such copy of $\pnt$ corresponds to
an interval $I = [\alpha,\beta]$, where $\alpha$ is $\min_{i \in
   V_\ell} i$, and $\beta$ is $\max_{i \in V_\ell} i$.  So $\pnt$ has
$h$ such intervals associated with it, say $I_1, \ldots, I_h$. We
generate $h$ new pairs from $\pnt$, namely, $Q(\pnt) = \brc{
   \pth[]{\pnt, I_1}, \ldots, \pth[]{\pnt, I_h}}$.

We set $\PntSet' = \cup_{\pnt} Q(\pnt)$, and $\FC{\Family} =
\brc{\FC{\range_i} \sep{ \range_i \in \Family}}$, where
\begin{equation}
    \FC{\range_i} = \brc{ \pth{\pnt, I} \in \PntSet' \sep{  \pnt \in
          \range_i \text{ and }  {i} \in I}}.
    \eqlab{r:i:definition}
\end{equation}
Note that $\cardin{\FC{\range_i}} = \cardin{\range_i}$, and it can be
interpreted as deciding, for each point $\pnt \in \range_i$, which one
of its copies should be included in $\FC{\range_i}$.

The following two claims follow easily from the construction.
\begin{claim}
    For the set cover instance defined by $(\PntSet', \FC{\Family})$
    there is a fractional solution of value $2 \sum_i x_i \le
    2\PriceF$.
    
    \clmlab{fractional:value}
\end{claim}

\begin{claim}
    An integral solution of value $\beta$ to the set cover instance
    $(\PntSet', \FC{\Family})$ implies a multi-cover to the original
    instance of cardinality at most $\beta$.
    
    \clmlab{integer:feasibility}
\end{claim}

We need the following easy lemma on the dimension of intersection of
two range spaces with bounded \VC dimension.
\begin{lemma}[\cite{h-gaa-08}]
    Let $\RangeSpace =(\GroundSet, \RangeSet)$ and $\RangeSpaceA =
    (\GroundSet, \RangeSetA)$ be two range spaces of \VC{}-dimension
    $\VCDim$ and $\VCDimA$, respectively, where $\VCDim,\VCDimA >
    1$. Let $\widehat{\RangeSet} = \brc{ \range \cap \rangeA \sep{
          \range \in \RangeSet, \rangeA \in \RangeSetA}}$.  Then, for
    the range space $\widehat{\RangeSpace} = (\GroundSet,
    \widehat{\RangeSet})$, we have that $\VCDim( \widehat{\RangeSpace}
    ) = O( \VCDim + \VCDimA)$.

    \lemlab{easy}
\end{lemma}

\begin{observation}
    If $\RangeSpace = (\GroundSet, \RangeSet)$ has \VC dimension
    $\delta$, and $\X \subseteq \RangeSet$, then the \VC dimension of 
    $(\GroundSet, \X)$ is bounded by $\delta$.

    \obslab{monotone}
\end{observation}

The crucial lemma is the following.
\begin{lemma}
    The \VC dimension of the set system $(\PntSet', \FC{\Family})$ is
    $O(\VCDim )$ and the \VC dimension of its dual set system is 
    $O(\VCDim^*)$.
    
    \lemlab{intersect}
\end{lemma}

\begin{proof}
    We define two set systems $\pth[]{\PntSet', \FA{\Family}}$ and
    $\pth[]{ \PntSet', \FB{\Family}}$ as follows.  $\FA{\Family} =
    \brc{\FA{\range_i} \sep{ \range_i \in \Family }}$ where
    \[
    \FA{\range_i} = \brc{ \pth{\pnt, I} \in \PntSet' \sep{ \pnt
          \in \range_i}},
    \]
    and $\FB{\Family} = \brc{\FB{\range_i} \sep{ \range_i \in \Family
       }}$, where
    $\displaystyle \FB{\range_i} = \brc{(\pnt, I) \in \PntSet' \sep{ i
          \in I}}$.

    Note that $\FC{\range_i} = \FA{\range_i} \cap \FB{\range_i}$ (see
    \Eqref{r:i:definition}). Therefore $(\PntSet', \FC{\Family})$ is
    formed by the intersection of ranges $(\PntSet', \FA{\Family})$
    with ranges of $(\PntSet', \FB{\Family})$.  Therefore the \VC
    dimension of $(\PntSet', \FC{\Family})$ is bounded by $O\pth{
       \FA{\VCDim} + \FB{\VCDim} }$ where $\FA{\VCDim}$ and
    $\FB{\VCDim}$ are the \VC dimensions of $(\PntSet', \FA{\Family})$
    and $(\PntSet', \FB{\Family})$ respectively, by \lemref{easy} and
    \obsref{monotone}. We observe that the set system
    $(\PntSet',\FA{\Family})$ has the same \VC dimension as that of
    $(\PntSet,\Family)$ since we only duplicate points.  The set
    system $(\PntSet',\FB{\Family})$ has constant \VC dimension
    $\FB{\VCDim} = 3$ since it is the intersection system of points on
    the line with intervals.


    The second part of the claim follows by a similar argument.
    Consider the dual range spaces of $\pth[]{\PntSet',
       \FA{\Family}}$, $\pth[]{\PntSet', \FB{\Family}}$, and
    $\pth[]{\PntSet', \FA{\Family}}$, respectively. The ground set of
    these range spaces can be made to be $\Family$. We have the
    following:
    \begin{itemize}
        \item $\FA{\Instance^*} = \pth[]{\Family, \FA{\X}} $, the range space
        dual to $\pth[]{\PntSet', \FA{\Family}}$, has for any point
        $(\pnt,I) \in \PntSet'$ a range that contains all the 
        $\range_i \in \Family$ that contains $\pnt$. It is therefore
        just the dual range space to $\Instance=(\PntSet, \Family)$,
        and it has \VC dimension $\delta^*$. 

        \item $\FB{\Instance^*} = \pth[]{\Family, \FB{\X}} $, the
        range space dual to $\pth[]{\PntSet', \FB{\Family}}$, for
        every $(\pnt,I) \in \PntSet'$, has the range containing all
        the sets $\range_i$ such that $i \in I$. As such,
        $\FB{\Instance^*}$ has a constant \VC dimension.

        \item $\FC{\Instance^*} = \pth[]{\Family, \FC{\X}} $, the
        range space dual to $\pth[]{\PntSet', \FC{\Family}}$, for
        every $(\pnt,I) \in \PntSet'$, has the range containing all
        the sets $\range_i$ such that $i \in I$ and $\pnt \in
        \range_i$.
    \end{itemize}
    We have that $\FC{\Instance^*}$ is the range space contained in the
    intersection of range spaces $\FA{\Instance^*}$ and
    $\FB{\Instance^*}$. \lemref{easy} and \obsref{monotone}
    imply that the \VC dimension of $\FC{\Instance^*}$ is $O(
    \delta^*)$.
\end{proof}

\bigskip

Now we apply the known results on the integrality gap of the \LP for
set cover as discussed in Section~\ref{sec:overview}.  These results
imply that for the set system
$(\PntSet',\FC{\Family})$ there is an integral set cover of value
$O(\VCDim^* \PriceF \log \PriceF)$ (here we use
\clmref{fractional:value} and \lemref{intersect}).  From
\clmref{integer:feasibility}, there is a multi-cover for the original
instance of the desired size. This completes the proof of the
theorem. We observe that the algorithm is in fact quite simple. After
solving the \LP, pick each range $\range_i$ independently with
probability $\min\{1, c x_i\}$ where $c = \alpha \cdot \VCDim^* \log
\PriceF$ for a sufficiently large constant $\alpha$. With constant
probability this yields a multi-cover.

\subsection{Multi-cover in spaces with bounded \TPDF{\VC}{VC} dimension
   when allowing repetition}
\seclab{with:reps}

We consider the case where sets in $\Family$ are allowed to be picked
multiple times to cover a point. For this purpose we use relative
approximations. The standard definition of relative approximation is
the dual of what we give below.

\medskip

Let $\alpha, \phi > 0$ be two constants. For a set system
$\Instance = (\PntSet, \Family)$, recall from \defref{depth} that
$\Cover{\pnt}{\Family}$ denotes
the number of sets in $\Family$ that contain the point $\pnt$.  A
\emphi{relative $(\alpha, \phi)$-approximation} is a subset $X
\subseteq \Family$ that satisfies
\begin{equation}
    (1-\alpha) \frac{\Cover{\pnt}{\Family}}{\cardin{\Family}} %
    \leq %
    \frac{\Cover{\pnt}{X}}{\cardin{X}} %
    \leq%
    (1+\alpha) \frac{\Cover{\pnt}{\Family}}{\cardin{\Family}}.
    \eqlab{rel-app}
\end{equation}
for each $\pnt \in \PntSet$ with $\Cover{\pnt}{\Family} \ge \phi \cdot
\cardin{\Family}$.  It is known \cite{lls-ibscl-01} that there exist
subsets with this property of size ${\displaystyle \frac{c\delta
   }{\alpha^2\phi}\log\frac{1}{\phi}}$, where $c$ is an absolute
constant, and $\delta$ is the \VC dimension of the dual set system
of $(\PntSet, \Family)$. Indeed, any random sample of that many sets from
$\Family$ is a relative $(\alpha, \phi)$-approximation with constant
probability. To guarantee success with probability at least $1-q$, one
needs to sample $\displaystyle \frac{c}{\alpha^2\phi}\pth{\delta
   \log\frac{1 }{ \phi} + \log\frac{1}{q}}$ elements of $X$, for a
sufficiently large constant $c$ \cite{lls-ibscl-01}.

\bigskip

To apply relative approximation for our purposes we let $N$ be a large
integer such that $Nx_i$ is an integer for each range $r_i$ (since the
$x_i$ are rational such an $N$ exists). We create a new set system
$(\PntSet, \Family')$ where $\Family'$ is obtained from $\Family$ by
duplicating each range $r_i \in \Family$ $Nx_i$ times. Thus
$\cardin{\Family'} = N f$. From the feasibility of $x$ for the \LP we
have that $\Cover{\pnt}{\Family'} \ge N \Demand{\pnt} \ge Nf
\Demand{\pnt}/f$ for each $\pnt \in \PntSet$.

Now we apply the relative approximation result to $(\PntSet,
\Family')$ with $\phi = 1/\PriceF$ and $\alpha = 1/2$ to obtain a set
$X \subset \Family'$ such that $\cardin{X} = \Theta(\VCDim^* \PriceF
\log \PriceF)$ and with the property that for each $\pnt \in \PntSet$,
\[
\frac{\Cover{\pnt}{\Family'}}{2 \cardin{\Family'}} \le
\frac{\Cover{\pnt}{X}}{\cardin{X}}.
\]
We have
\[
\Cover{\pnt}{X}%
\geq%
\frac{\cardin{X}}{2} \cdot \frac{\Cover{\pnt}{\Family'}}{
   \cardin{\Family'}}%
\geq%
\frac{\cardin{X}}{2} \cdot \frac{N \Demand{\pnt} }{N \PriceF}
= \Demand{\pnt} \cdot \Omega\pth{  \VCDim^* \log \PriceF }
\geq 
\Demand{\pnt},
\]
as desired.


Note that $X$ is picked from $\Family'$ which has duplicate copies of
sets from $\Family$.  Recall that the algorithm, from the previous
section (which is for the variant without repetition), picks each
range $\range_i$ independently with probability $\min\{1,c x_i \cdot
\VCDim^* \log \PriceF\}$; and this yields a feasible multi-cover
\emph{without} repetitions.  It may be possible to analyze this
algorithm (i.e., without repetitions) directly by a careful
walkthrough of the proof for relative approximations.



\section{Multi-cover for Halfspaces in 3d and Generalizations}

In this section, we show that improved approximations can be obtained
for specific classes of set systems induced by geometric shapes of low
complexity. In particular, we describe an $O(1)$ approximation for the
multi-cover problem when the points are in $\Re^3$, and the ranges are
induced by halfspaces. The main idea, of using cuttings, extends also to
other nice shapes. We outline the extensions and some applications in
\secref{generalization}.

\subsection{Total demand, Sampling, and Residual demand}

We develop some basic ingredients that are useful in randomly
rounding the \LP solution. These ingredients apply to a generic
multi-cover instance, not necessarily a geometric one, however we use
the notation of points and ranges for continuity.

\begin{lemma}
    Given a multi-cover instance $\Instance = (\PntSet,
    \Family)$, one can compute a cover for $\Instance$ of size
    no more than the total demand $\DemandExt{\Instance}{\PntSet}$.
    
    \lemlab{energy:small}
\end{lemma}

\begin{proof}
    Indeed, scan the unsatisfied points of $\PntSet$ one by one, and
    for each such point $\pnt$, add to the solution $\Demand{\pnt}$ ranges
    that cover it, picked arbitrarily.
    Clearly, the ranges that are picked satisfy all the demands, and
    the number of ranges picked is at most $\sum_{\pnt}\Demand{\pnt} =
    \DemandExt{\Instance}{\PntSet}$.
\end{proof}

\bigskip

Given an instance of multi-cover $\Instance = (\PntSet, \Family)$ and
a feasible fractional solution $x$, a \emphi{$cx$-sample} for a scalar
$c$ is a random sample of $\Family$, formed by independently picking
each of the ranges $\range_i \in \Family$ with probability $\min\{1, c
x_i\}$, where $x_i$ is the value assigned to $\range_i$ by the
fractional solution. (For the $i$ with $cx_i \geq 1$, so that $i$ is
chosen with probability one, we will simply assume that such choices
have been made, and the demand removed; that is, we assume that
hereafter that $x_i \leq 1/c$.  Since the number of such $i$ is at
most $c\,\PriceF$, this step does not affect our goal of obtaining an
output cover with $O(\PriceF)$ sets.)

\begin{lemma}
    Let $c\geq 4$ be a constant and let $\Instance = (\PntSet,
    \Family)$ be a multi-cover instance with an \LP solution
    satisfying $x_i\le 1/c$ for all $i$. If $\RSample$ is a
    $cx$-sample and $\pnt \in \PntSet$ is a point with demand
    $\DemandChar = \Demand{\pnt}$, then
    \[
    \Prob{ \VBig{0.5cm}\pnt \text{ is not fully covered by } \RSample
    } = \Prob{ \VBig{0.5cm} \Cover{\pnt}{\RSample} < \DemandChar} \leq
    \exp \pth{ - \frac{c }{4} \DemandChar },
    \]
    and
    $\displaystyle \Ex{\VBig{0.5cm} \REnergy{\pnt, \RSample} } \leq
    \exp \pth{-\frac{c}{4} \DemandChar}$.
    
    \lemlab{demand}
\end{lemma}

\begin{proof}
    Let $X_i$ be the indicator variable which is equal to one if the
    $cx$-sample includes the range $\range_i \in \Family$, and is zero
    otherwise. Let $Y = \Cover{\pnt}{\RSample} = \sum_{i: \pnt \in
       \range_i} X_i$; observe that $\mu = \Ex{ Y } \geq c\,
    \DemandChar$ using the facts that $x$ is a feasible solution to
    \LP, and $x_i \le 1/c$ for all $i$. For $j \in [0, \DemandChar]$,
    we apply the Chernoff inequality \cite{mr-ra-95} and use the fact
    that $c \geq 4$ to obtain:
    \begin{eqnarray*}
        \Prob{\VBig{0.5cm} \Cover{\pnt}{\RSample} \leq \DemandChar - j}
        &\leq& \Prob{ \VBig{0.5cm}Y <
           \mu \pth{1 - (c-1)/c - j/\mu} } \leq \exp \pth{ - \frac{\mu}{2} \pth{
              \frac{c-1}{c} + \frac{j}{\mu} }^2 }\\
        &\leq& \exp \pth{ - \frac{\mu}{4} - \frac{3}{4}j} \leq
        \exp \pth{ - \frac{c }{4} \DemandChar - \frac{3}{4}j}. 
    \end{eqnarray*}
    The first statement of the lemma follows by substituting $j=1$
    and observing that the desired bound follows,
    and the second follows by using the fact that, for a random
    variable $Z$ taking non-negative integral values, that $\Ex{Z} =
    \sum_{k > 0} \Prob{Z \geq k}$.  This implies
    \begin{eqnarray*}
        \Ex{\VBig{0.5cm} \REnergy{\pnt, \RSample} }
        &=&  \sum_{1\le j\le \DemandChar} \Prob{\VBig{0.5cm} 
           \Cover{\pnt}{\RSample} \leq \DemandChar - j}
        \leq%
        \sum_{1\le j\le \DemandChar} \exp \pth{ - \frac{c }{4} 
           \DemandChar - \frac{3}{4}j}
        \\ &=&%
        \exp \pth{ - \frac{c }{4}\DemandChar } 
        \sum_{1\le j\le \DemandChar} \exp \pth{- \frac{3}{4}j}
        \leq%
        \exp \pth{ - \frac{c }{4}\DemandChar } \frac{1}{\exp \pth{3/4} -
           1}
        \leq         \exp \pth{ - \frac{c }{4}\DemandChar },
    \end{eqnarray*}
    as claimed.
\end{proof}

\bigskip

In the following, for $t \geq 1$, let
\[
\PntSet_t = \brc{\pnt \in \PntSet \sep{ \; t \leq \Demand{\pnt} < 2t \;}}.
\]

The lemma below implies that if the number of points in the set system
is ``small'' then the multi-cover problem can almost be solved in one
round of sampling.

\begin{lemma}
    Suppose there is a probability distribution on a collection of
    multi-cover instances such that an instance $\Instance = (\PntSet,
    \Family)$ chosen from the distribution satisfies, for any $t \geq
    1$, that
    \[
    \Ex{ \VBig{0.5cm} \cardin{\PntSet_t}\;} \leq V \cdot K^t,
    \]
    where $K$ and $V$ are fixed parameters of the distribution.  Then
    there is a value $c$ depending on $K$, such that for any feasible
    fractional solution $x$ to $\Instance$, a $cx$-sample $\RSample$
    results in expected total residual demand $\REnergy{\PntSet,
      \RSample} \leq V$; here the expectation is with respect to the
    randomness of $\Instance$ and the independent randomness of the
    $cx$-sample.
    
    \lemlab{small:set}
\end{lemma}

\begin{proof}
    Let $\RSample$ be a $cx$-sample of $\Family$ for fixed $c\geq
    4+4\log K$. Let $X$ be the subset of $\Family$ with all ranges
    having $x_i \geq 1/c$.  Since $\RSample \setminus X$ is also a
    $cx$-sample of $\Instance\setminus X$, we assume hereafter that
    $X$ is empty; the result for general $X$ follows by application of
    the result to $\Instance\setminus X$.
  
    By applying \lemref{demand} to the induced range space $\pth[]{
       \PntSet_t, \Family}$, we have
    \begin{align*}
        \Exs{\Instance,\RSample}{\VBig{0.5cm} \REnergy{\PntSet_t, \RSample} }
        &\leq \Exs{\Instance}{\sum_{\pnt\in\PntSet_t} 
           \Exs{\RSample}{\REnergy{\pnt, \RSample}} }
        \leq \Exs{\Instance}{\cardin{\PntSet_t} \exp\pth{-\frac{c}{4}
              t}}
        = \Exs{\Instance}{\cardin{\PntSet_t} \MakeBig\!}
        \exp\pth{-\frac{c}{4} t}
        \\
        &\leq   V K^t \exp\pth{-\frac{c}{4} t}
        \leq V \exp\pth{-t (c/4 - \log K)}
        \leq   V \exp\pth{-t } .
    \end{align*}
    Then, by linearity of expectation, we have
    \begin{eqnarray*}
        \Ex{\VBig{0.5cm} \REnergy{\PntSet, \RSample} }
        =  \sum_{i=0}^{\infty} 
            \Ex{\VBig{0.5cm} \REnergy{\PntSet_{2^i}, \RSample} }
        \leq \sum_{i=0}^\infty V \exp\pth{-2^i}
        \leq V.
    \end{eqnarray*}
    Thus, after $cx$-sampling, the residual instance has total expected
    demand bounded by $V$, as claimed. 
\end{proof}


\subsection{Clustering the given instance}

The key observation to solve the multi-cover problem in our settings
is \lemref{small:set}, as it provides a sufficient condition for an
$O(1)$ approximation.  Of course, it might not be true (even in low
dimensional geometric settings) that the number of points (i.e., the
total residual demand) is small enough, as required to apply this
lemma. We preprocess the given instance via an initial sampling step
and then employ a clustering scheme that partitions the points into
regions; we argue that these regions and an induced multi-cover
instance on them satisfies the conditions of the lemma.

The \emphi{depth} of a simplex $\Simplex$ in a set of weighted
halfspaces is the minimum depth of any point inside $\Simplex$, see
\defref{depth}.

To perform the aforementioned clustering, we will use the shallow
cutting lemma of \Matousek \cite{m-rph-92}. We next state it in the
form needed for our application, which is a special case of
\thmref{shallow:cutting}.

\medskip

\begin{lemma}
    Given a set $\Family$ of weighted halfspaces in $\Re^3$, with
    total weight $W$, there is a randomized polynomial-time
    algorithm that generates a set $\Gamma$ of simplices, called
    a \emphi{$(1/4W)$-cutting}, with the following properties:
    the union of the simplices covers $\Re^3$;
    the total weight of the boundary planes of $\Family$ intersecting any
    simplex of $\Gamma$ is bounded by $1/4$; and finally, for any
    $t\geq 0$, the expected total number of simplices of depth at most
    $t$ is $O\pth{ W  t^{2} }$.  (Here the expectation
    is with respect to the randomness of the algorithm.)
    \lemlab{decomposition}
\end{lemma}

\subsubsection{The algorithm}

Given an instance of multi-cover $\Instance = (\PntSet,\Family)$ of
points and halfspaces in $\Re^3$, our algorithm first computes the
fractional solution to the \LP induced by $\Instance$, yielding
weights $x_i$.  Next, for $\beta$ an absolute constant in $(0,1/4)$ to
be specified later, we put in the set $X$ all the ranges $r_i$ with
$x_i\geq\beta$.  Let $\pth[]{ \PntSet', \Family'} = \pth{\PntSet,
   \Family} \setminus X$. Let $\PriceF' = \sum_{\range_i \in \Family
   \setminus X} x_i$ be the total weight of the remaining ranges.

The remainder of the algorithm uses a auxiliary abstract multi-cover
instance derived using cuttings, as described next.

Using the weights $x_i$, we build a $(1/4\PriceF')$-cutting $\Gamma$
for $\Family'$. This induces an abstract multi-cover instance
$(\Gamma, \Family')$, where a simplex $\Delta\in\Gamma$ is covered by
halfspace $h\in\Family'$ only if the interior of $\Delta$ is contained
inside $h$ and it does not meet the boundary plane of $h$.  The demand
$\Demand{\Delta}$ is defined to be $\max_{\pnt\in \PntSet\cap \Delta}
\REnergy{\pnt, \Family'}$.

A feasible solution to $(\Gamma, \Family')$ is also, by construction,
a feasible solution for the original instance
$\Instance$. Furthermore, any feasible fractional solution for
$\Instance$ can be transformed into a feasible fractional solution for
$(\Gamma, \Family')$, at the cost of a constant factor.  Indeed, the
weights $x_i$ give a feasible fractional solution to
$\Instance\setminus X$, and so the depth of $\Delta$ is at least
$\Demand{\Delta} - 1/4$, where $\Delta$ ``loses'' at most weight $1/4$
of depth due to halfspaces whose boundary planes cut $\Delta$.  It
follows that if the depth is measured with respect to weights $\hat
x_i = 2x_i$, the new depth of $\Delta$ (i.e., the point with minimum
cover in $\Delta$) is at least $2\Demand{\Delta} - 1/2 >
\Demand{\Delta}$.  That is, the weights $\hat x_i$ give a feasible
fractional solution to the multi-cover instance $(\Gamma, \Family')$.
Note that since $\beta < 1/4$, the weights $\hat x_i$ satisfy $\hat
x_i < 1$, for all $i$.

The remainder of the algorithm is to apply the approach implied by
\lemref{small:set}: we find a $c\hat x$-sample $\RSample$,
with $c$ to be determined; this induces a
residual multi-cover problem $(\Gamma, \Family')\setminus\RSample$,
which we solve using the simple technique of \lemref{energy:small}.
Letting $U$ denote the resulting combined solution to $(\Gamma, \Family')$, we
return $U\cup X$ as a cover for the original multi-cover problem.

The analysis of this algorithm is the proof of the following result.

\begin{theorem}
    Let $\Instance = (\PntSet, \Family)$ be an instance of multi-cover
    formed by a set $\PntSet$ of points in $\Re^3$, and a set
    $\Family$ of halfspaces. Then, one can compute, in randomized
    polynomial time, a subset of halfspaces of $\Family$ that meets
    all the required demands, and is of expected size
    $O\pth{\PriceF}$, where $\PriceF$ is the value of an optimal
    fractional solution to \LP.

    \thmlab{main}
\end{theorem}

\begin{proof}
  We described the algorithm above, except for the values of $c$ and
  $\beta$.
    
  By \lemref{decomposition}, the expected number of simplices in the
  cutting $\Gamma$ of demand at most $t$ is $O\pth{ W t^2 }$, where $W
  = \PriceF' \le \PriceF(\Instance)$, which implies that
  \lemref{small:set} can be applied, with $V = \PriceF(\Instance)$,
  $K$ an absolute constant, and using the weights $\hat x_i$.  Since
  $\sum_i \hat x_i \le 2\PriceF(\Instance)$, the expected size of $U$
  is at most $(c+2)\PriceF(\Instance)$, using the absolute constant
  value of $c$ used in this application of \lemref{small:set}.
  Observing that $\cardin{X} \leq \PriceF\pth{\Instance}/\beta$, and
  taking $\beta = 1/2c$ to allow the $cx$-sample probabilities $c\hat
  x_i$ to be less than one, we have that the returned solution $U\cup
  X$ to $\Instance$ has expected cardinality at most
  $(3c+2)\PriceF(\Instance)$, which is $O(\PriceF(\Instance))$.
    
  The only non-trivial step in terms of verifying the running time is
  for computing the cutting and \lemref{decomposition} guarantees the
  running time.
\end{proof}

\begin{remark}
    The shallow-cutting lemma (\lemref{decomposition}) is shown via a
    random sampling argument, and our rounding algorithm is also based
    on random sampling, given the cutting as a black-box. One could do
    a direct analysis of random sampling by unfolding the proof of the
    cutting lemma. However, the indirect approach is easier to see and
    highlights the intuition behind the proof.
\end{remark}

\section{Generalizations and Applications}
\seclab{generalization}

We now examine to what extent the result derived for covering points
in $\Re^3$ by halfspaces generalizes to other shapes. 

\subsection{Well behaved shapes}

We are interested in set systems $(\PntSet, \Family)$ where $\Family$
is a set of ``well-behaved'' shapes such as disks or fat triangles.
As we remarked already, it is shown in
\cite{cv-iaags-07} that the existence of good $\eps$-nets for such
shapes can be derived from bounds on their union complexity. For
example, it is shown that if $\Family$ is a set of \emph{fat}
triangles in the plane then there is an $O(\log \log \PriceF)$
approximation for the set cover problem. For fat wedges one obtains an
$O(1)$ approximation. Here we show that union complexity bounds can be
used to derive approximation ratios for the multi-cover problem that
are similar to those derived in \cite{cv-iaags-07} for the set cover
problem. Following the scheme for halfspaces, the key tool is the
existence of shallow cuttings. To this end we describe some general
conditions for the shapes of interest and then state a shallow cutting
lemma.

Let $\Family$ be a set of $n$ shapes in $\Re^d$, such that their union
complexity for any subset of size $r$ is (at most) $\FZero(r)$, for
some function $\FZero(r) \geq r$. Similarly, let $O\pth{r^{d}}$ be the
upper bound on total complexity of an arrangement of $r$ such shapes.

Let $X$ be a subset of $\Re^d$. We assume that given a subset
$\FamilyB \subseteq \Family$, one can perform a decomposition the
faces of the arrangement $\ArrX{\FamilyB}$ that intersects $X$ into
cells of constant descriptive complexity (e.g., vertical trapezoids),
and the complexity of this decomposition is proportional to the number
of vertices of the faces of $\ArrX{\FamilyB}$ that intersects $X$.
Finally, we assume that the intersection of $d$ shapes of $\Family$
generates a constant number of vertices.

One can then derive the following version of \Matousek's shallow
cutting lemma.  We emphasize that this lemma is a straightforward (if
slightly messy) adaption of the result of \Matousek.  A proof is
sketched in \apndref{shallow:cutting}.

\begin{theorem}
    Given a set $\Family$ of ``well-behaved'' shapes in $\Re^d$ with
    total weight $n$, and parameters $r$ and $k$, one can compute a
    decomposition of space into $O(r^d)$ cells of constant descriptive
    complexity, such that total weight of boundaries of shapes of
    $\Family$ intersecting a single cell is at most $n/r$.
    Furthermore, the expected total number of cells containing points
    of depth smaller than $k$ is 
    \[
    O\pth{ \pth{\frac{r k}{n} + 1}^d \FZero\pth{\frac{n}{k}}},
    \]
    where $\FZero\pth{\ell}$ is the worst-case combinatorial complexity of
    the boundary of the union of $\ell$ shapes of $\Family$.
    
    \thmlab{shallow:cutting}
\end{theorem}

Using the same scheme as that for halfspaces we can derive
approximation ratios for the multi-cover problem for shapes that have
the property that $\FZero\pth{n}$ is near-linear in $n$. An
approximation ratio of $O(\FZero\pth{ \Opt}/\Opt)$ easily follows, but
in fact, by using the oversampling idea of Aronov \etal
\cite{aes-ssena-09}, we can improve this to $O(\log (\FZero\pth{
  \Opt}/\Opt ))$. We use the shallow cutting lemma as a black box, and
hence our argument is arguably slightly simpler than then one in
\cite{aes-ssena-09} and our result can be interpreted as a
generalization.

\begin{theorem}
  Let $\Instance = (\PntSet, \Family)$ be an instance of multi-cover
  formed by a set $\PntSet$ of points in $\Re^d$, and a set $\Family$
  of ranges.  Furthermore, the union complexity of any $\ell$ such ranges
  is (at most) $\FZero(\ell)$, for some function $\FZero(\ell) \geq \ell$.
  Then, one can compute, in randomized polynomial time, a subset of
  ranges of $\Family$ that meets all the required demands, and is of
  expected size $O\pth{ \PriceF \log \frac{\FZero\pth{ \PriceF
      }}{\PriceF}}$, where $\PriceF$ is the value of an optimal
  fractional solution to \LP.

    \thmlab{main:2}
\end{theorem}
\begin{proof}
    As before, we compute the \LP relaxation, and take all the ranges
    that the value of $x_i \geq \beta$, where $\beta = \alpha/ \log
    \frac{\FZero\pth{ \PriceF }}{\PriceF}$ for some sufficiently small
    constant $\alpha$.  Next, we compute a $(1/4\PriceF)$-cutting
    $\Gamma$ of residual system $\pth[]{ \PntSet', \Family'}$. Using
    \thmref{shallow:cutting} with parameters $r= 4\PriceF$, $n =
    \PriceF$ and $k=t$, there are at most
    \[   
    O \pth{ (t+1)^{d} \;\PriceF \; 
       \frac{\FZero\pth{ \PriceF }}{\PriceF}} 
    \]
    cells, with depth at most $t$. In particular, this bounds the
    number of cells in the cutting with depth in the range $t-1$ to
    $t$.  We pick a random sample $\RSample$ of (expected) size $h =
    O\pth{ \PriceF \log \frac{\FZero\pth{ \PriceF }}{\PriceF}}$ from
    $\Family'$, by performing a $cx$-sample from $\Family'$, where $c =
    O \pth{ \log \frac{\FZero\pth{ \PriceF }}{\PriceF}}$. Arguing as
    in \lemref{demand}, the expected residual demand for a cell of
    $\Gamma$ with demand $t$ is $t \exp \pth{ -c 
       t/4}$. Therefore, the expected total residual demand in $(\Gamma,
    \Family') \setminus \RSample$ is
    \[
    O\pth{ \sum_{t=1}^\infty \exp \pth{ -\frac{c }{4} t} (t+1)^{d+1}
       \; \PriceF\; \frac{\FZero\pth{ \PriceF }}{\PriceF}} = O \pth{
       \PriceF}.
    \]
    Using \lemref{energy:small}, the residual multi-cover instance
    $(\Gamma, \Family') \setminus \RSample$ has a cover of expected
    size $O(\PriceF)$.  Thus, we have shown that the original
    multi-cover instance has a cover of expected size $O(\PriceF/\beta
    + h + \PriceF)=O\pth{ \PriceF \log \frac{\FZero\pth{ \PriceF
        }}{\PriceF}}$.
\end{proof}

\paragraph{Applications:} The above general result can be combined
with known bounds on $\FZero\pth{n}$ to give several new results. We
follow \cite{cv-iaags-07,aes-ssena-09} who gave approximation ratios
for the set cover problem using a similar general framework; we give
essentially similar bounds for the multi-cover problem. All the
instances below involve shapes in the Euclidean plane.
\begin{itemize}
    \item $O(1)$ approximation for pseudo-disks, fat triangles of
    similar size, and fat wedges.

    \item $O(\log \log \log \PriceF)$ approximation for fat triangles
    (which also implies similar bounds for fat convex polygonal shapes
    of constant description complexity).

    \item $O(\log \alpha(\PriceF))$ approximation for regions each of
    which is defined by the intersection of the non-negative $y$
    halfplane with a Jordan region such that each pair of bounding
    Jordan curves intersecting at most three times (not counting the
    intersections on the $x$ axis). Here $\alpha(n)$ is the inverse
    Ackerman function.
\end{itemize}

\subsection{Unit Cubes in 3d}
We also get a similar result for the case of axis-parallel unit cubes.

In \cite{cv-iaags-07} an $O(1)$ approximation is also shown that for
the problem of covering points by unit sized axis parallel cubes in
three dimensions. There is a technical difficulty for this case.
Although it is known from \cite{bsty-vdhdc-98} that the combinatorial
complexity of the union of $n$ cubes is $O(n)$, the same bound is not
known for the canonical decomposition of the exterior of the union as
required by our framework. The same difficulty is present in
\cite{cv-iaags-07} and they overcome this by taking advantage of the
fact that all cubes are unit sized. The basic idea is to use a grid
shifting trick to decompose the given instance into independent
instances such that each instances has cubes that contain a common
intersection point. For this special case one can show that the
canonical decomposition of the exterior of the shapes has linear
complexity. This suffices for the framework in \cite{cv-iaags-07}.
For our framework we need a cutting.

\newcommand{\rect}{\mathsf{r}}
\newcommand{\face}{f}
\newcommand{\faceA}{g}
\newcommand{\cube}{c}

\begin{lemma}
    Let $S$ be a set of $n$ axis-parallel unit cubes in three
    dimensions, all of them containing (say) the origin. Then, one can
    decompose the arrangement of $\ArrX{S}$ into a canonical
    decomposition of axis parallel boxes, such that the complexity of
    decomposing every face is proportional to the number of vertices
    on its boundary.

    \lemlab{shallow:cubes}
\end{lemma}

\begin{proof}
    First we break the arrangement into eight octants by the three
    axis planes ($xy$, $yz$ and $xz$ planes).  We will describe how to
    decompose the arrangement in the positive octant, and by symmetry
    the construction would apply to the whole arrangement.

    So, let $\face$ be a 3d face of the arrangement (when clipped to
    the positive octant). Let $I$ be the cubes of $S$ that contain
    $\face$, and similarly, let $B$ be the set of cubes of $S$ that
    contribute to the boundary of $\face$, but do not include $\face$
    in their interior. As such, we have that
    \[
    \face = \mathrm{closure}\pth{ \pth{\bigcap_{\cube \in I} \cube}
       \;\setminus \; \pth{\bigcup_{\cube' \in B} \cube'}}.
    \]
    (If the set $I$ is empty, we will add a fake huge cube to ensure
    $f$ is bounded.) Now, the first term is just an axis-parallel
    box. Intuitively, the second term (the ``floor'' of $\face$) is a
    (somewhat bizarre) collection of ``stairs''. Note, that any
    vertical line that intersects $\face$, intersects it in an
    interval. In particular, let $\faceA$ the top face (in the $z$
    direction) of $\face$, and observe that, since all the cubes of
    $S$ contain the origin, it must be that any line that intersect
    $\face$ must also intersect $\faceA$. As such, let us project all
    the edges and vertices of $\face$ upward till the hit
    $\faceA$. This results in a collection $W$ of (interior) disjoint
    segments that partition (the rectangular polygon) $\face$. We
    perform a vertical decomposition of the planar arrangement formed by
    $\ArrX{W}$ (including the outer face of this arrangement, which is
    $\faceA$). This results in $O\pth{ \cardin{\face}}$ collection of
    (interior) disjoint rectangles that cover $\faceA$, where
    $\cardin{\face}$ is the number of vertices on the boundary of
    $\face$. Furthermore, for such a rectangle $\rect$, there is no
    edge or vertex of $\face$, such that their vertical projection
    lies in the interior of $\rect$. Namely, we can erect a vertical
    prism for each face of the vertical decomposition of $\ArrX{W}$,
    till the prism hits the bottom boundary of $\face$. This result in
    a decomposition of $\face$ into $O(\cardin{\face})$ disjoint
    boxes, as required.
\end{proof}
\medskip

\lemref{shallow:cubes} implies that an the arrangement $\ArrX{S}$, can
be decomposed into (canonical) boxes, in such a way that the number of
boxes of certain depth $t$, is proportional to the number of vertices
of $\ArrX{S}$ of this depth. This implies that we can apply the
shallow cutting lemma to $S$ (we remind the reader that all the
axis-parallel unit cubes of $S$ contain the origin).

This is sufficient to imply $O(1)$ approximation to
multi-cover. Indeed, let $\Instance = \pth{\PntSet, \Family}$ be the
given instance of multi-cover, where $\Family$ is a set of unit-cubes
in three dimensions. Let $G$ be the unit grid, and for any point $\pntA \in G$, let
$\Family_\pntA$ be the set of cubes of $\Family$ that contains $\pnt$
(for the simplicity of exposition, we assume that every cube of
$\Family$ is contained in exactly one such set, as this can be easily
guaranteed by shifting $G$ slightly). Next, solve the \LP associated
with $\Instance$, and associate a point $\pnt \in \PntSet$ with $\pntA
\in G$, if the depth of $\pnt$ in $\Family_\pntA$ is at least $1/8$
(if $\pnt$ can be associated with several such instances, we pick the
one that provides maximum coverage for $\pnt$). Let $\PntSet_\pntA$ be
the resulting set of points. Thus, for any point in $\pntA \in G$,
there is an associated instance of multi-cover $\pth{\PntSet_\pntA,
   \Family_\pntA}$. Clearly, a constant factor approximation for each
of these instances, would lead to a constant factor approximation for
the whole problem.

Now, $\Family_\pntA$ is made of cubes all containing a common point,
and as such \lemref{shallow:cubes} implies that shallow cutting would
work for it. In particular, we can now apply the algorithm of
\thmref{main} to this instance, and get a constant factor
approximation (here, implicitly, we also used the fact that the union
complexity of $n$ axis-parallel unit cubes is linear).    This
implies the following theorem.

\begin{theorem}
    Let $\Instance = (\PntSet, \Family)$ be an instance of multi-cover
    formed by a set $\PntSet$ of points in $\Re^3$, and a set
    $\Family$ of axis-parallel unit cubes. Then, one can compute, in
    randomized polynomial time, a subset of cubes of $\Family$
    that meets all the required demands, and is of expected size
    $O\pth{\PriceF}$, where $\PriceF$ is the value of an optimal
    fractional solution to \LP.

    \thmlab{cubes}
\end{theorem}

\section{Conclusions}
We presented improved approximation algorithms for set multi-cover  in
geometric settings. Our key insight was to produce a ``small'' instance
of the problem by clustering the given instance. This in turn was
done by using a variant of shallow cuttings. We believe that this
approach might be useful for other problems in geometric settings. 

An interesting open problem, is to obtain improved algorithms for the
set cover and the set multi-cover problems in geometric settings when
the sets/shapes have costs associated with them and the goal is to
find a cover of lowest cost. Can the results from
\cite{c-apca-93,bg-aoscf-95, cv-iaags-07} and this paper be extended
to this more general setting?

Recently, Mustafa and Ray \cite{mr-pghsp-09} gave a PTAS for the
problem of covering points by disks in the plane; their algorithm is
based on local search. It would be interesting to see if this
algorithm can be adapted to the multi-cover problem.




\NotSarielComp{\bibliographystyle{alpha}}%
\SarielComp{\bibliographystyle{salpha}}%
\bibliography{multi_cover}


%

\appendix

\section{A shallow cutting lemma for ``nice'' shapes}
\apndlab{shallow:cutting}

In this section, we prove \thmref{shallow:cutting}, a variant of the
shallow cutting lemma of \Matousek in a slightly different setting.
We include the details for the sake of completeness, which are not
hard in light of \Matousek's work \cite{m-rph-92}. Our description is
somewhat informal, for simplicity. The family of shapes that we
consider needs to satisfy the assumptions outlined in
\secref{generalization}.

\paragraph{Building $(1/r)$-cuttings.}
When computing cuttings, one first picks a random sample $\RSample$ of
size $r$ of the objects of $\Family$, and computes the decomposition
$\ArrVD{\RSample}$ of the arrangement of the random sample.  For a cell
$\Cell$ in this decomposition, let $\CList{\Cell}$ be the list of
shapes of $\Family$ whose boundaries intersect the interior of
$\Cell$.  If $\cardin{ \CList{\Cell}} \leq n/r$ then it is acceptable,
and we add it to the resulting cutting.

Otherwise, we need to do a local patching up, by partitioning each such
cell further. Specifically, let $t_\Cell = \ceil{ \CList{\Cell} /
   \pth{ n/r} }$ be the \emphi{excess} of $\Cell$. We take a random
sample $\RSample_\Cell$ of size $O(t_\Cell \log(t_\Cell))$ from
$\CList{\Cell}$. With constant probability, this is a $1/t_\Cell$-net
of $\CList{\Cell}$ (for ranges formed by our decomposition).
We verify that it is such a net, and if not, we resample, and repeatedly
do so until we obtain a $1/t_\Cell$-net. To do the verification,
we build the arrangement of $\RSample_\Cell$ inside $\Cell$, and
compute its decomposition, and check that all the cells in this decomposition
intersect at most $n/r$ boundaries of the shapes of $\Family$. Let
$\Decomp{\Cell}$ denote this decomposition of $\Cell$ (if $\Cell$ has
excess at most $1$, then we just take $\Decomp{\Cell}$ to be $\brc{\Cell}$).
Clearly, the set
\[
\bigcup_{\Cell \in \ArrVD{\RSample}} \Decomp{\Cell}
\]
forms a decomposition of $\Re^d$ into regions of constant complexity,
and each region intersects at most $n/r$ boundaries of the shapes of
$\Family$.

It is well known that the complexity of the resulting cutting is (in
expectation) $O( r^d )$ \cite{cf-dvrsi-90}

Let $\Cutting$ denote the resulting cutting.

\paragraph{Size of cutting at a certain depth.}

Here we are interested in the number of cells in the arrangement
$\ArrVD{\RSample}$ that cover ``shallow'' portions of
$\ArrX{\Family}$.  Formally, the \emphi{depth} of a point $\pnt \in
\Re^d$, is the number of shapes of $\Family$ that cover it. Let
$\FAtMost{k}{n}$ denote the maximum number of vertices of depth at
most $k$ in an arrangement of $n$ shapes.  Clarkson and Shor
\cite{cs-arscg-89} showed that $\FAtMost{k}{n} = O\pth{ k^d
   \FZero(n/k)}$.  Specifically, we are interested in the number of
cells of $\Cutting$ that contain points of depth at most $k$.  The
\emphi{$k$\th level} is the closure of all the points on the boundary
of the shapes that are contained inside $k$ shapes.

Now, the expected number of vertices of $\ArrX{\RSample}$ that are of
depth at most $k$ in $\ArrX{\Family}$ is
\[
O\pth{\pth{\frac{r}{n}}^d k^d \FZero(n/k)} = O\pth{ \pth{\frac{rk}{n}}^d
   \FZero(n/k)},
\]
since for a given vertex of $\ArrX{\Family}$ of depth at most $k$, the
probability that all $d$ shapes that define it will picked to be in
$\RSample$ is $O\pth{(r/n)^d}$.  This unfortunately does not bound the
number of cells in the decomposition of $\ArrVD{\RSample}$ that
contain points of depth at most $k$, since we might have cells that
cross the $k$\th level.

So, let $X \subseteq \Re^d$ be a fixed subset of space, and let
$x(\cardin{\RSample})$ be the number of cells of $\ArrVD{\RSample}$
that intersect $X$. Let $x(r)$ denote the maximum value of
$x(\cardin{\RSample})$ over all samples $\RSample$ of size $r$.
Similarly, let $x^t(\RSample)$ denote the number of cells in
$\ArrVD{\RSample}$ that intersect $X$ and have excess more than $t$
(i.e., there are at least $t \cdot n/r$ shapes intersecting this
cell). 


Chazelle and Friedman \cite{cf-dvrsi-90} showed an exponential decay
lemma stating that $\Ex{ x^t(\RSample)} = O\pth{2^{-t}
   \Ex{x(\RSample)}}$.  We comment that, in fact, one can prove
directly from the Clarkson-Show technique a polynomial decay lemma,
which is sufficient to prove the shallow-cutting lemma. This
polynomial decay lemma is implicit in the work of d{}e Berg and
Schwarzkopf \cite{bs-ca-95} although it was not stated explicitly (it
also made a stealthy appearance in Clarkson and Varadarajan work
\cite{cv-iaags-07}, but \cite{bs-ca-95} seems to be the earliest
reference).

\begin{lemma}[Polynomial decay lemma.]
    For $t \geq 1$, let $\RSample$ be a random sample of size $r$ from
    $\Family$, and let $c \geq 1$ be an arbitrary constant. Then $\Ex{
       x^t(\RSample)} = O\pth{x(r) / t^c }$.
\end{lemma}
\begin{proof}
    By the Clarkson-Shor technique \cite{cs-arscg-89, c-arscg-88}, we
    have that 
    \[
    \Ex{ \sum_{\Cell \in \ArrVD{\RSample}} \cardin{ \CList{\Cell}}^c }
    = O\pth{ \pth{ \frac{n}{r} }^c \Ex{x(\RSample)}}
    = O\pth{ \pth{ \frac{n}{r} }^c x(r)}.
    \]
    In particular, if there are $ x^t(\RSample)$ cells in
    $\ArrVD{\RSample}$ with conflict-list of size larger than
    $t(n/r)$, then they contribute to the left size of the above
    equation the quantity $    x^t(\RSample) (t(n/r))^c$. We conclude
    that 
    \[
    \Ex{ x^t(\RSample)  (t(n/r))^c} =
    O\pth{ \pth{ \frac{n}{r} }^c x(r)},
    \]
    which implies that $\Ex{ x^t(\RSample) } = O\pth{ x(r)/t^c}$, as
    claimed.    
\end{proof}

\begin{lemma}
    The expected number of cells in the $(1/r)$-cutting $\Cutting$ of
    $\Family$ that contain points of depth at most $k$ is bounded by
    \[
    O\pth{ \pth{\frac{r k}{n} + 1}^d \FZero\pth{\frac{n}{k}}}.
    \]
\end{lemma}

\begin{proof}
    If a cell $\Cell$ of $\ArrVD{\RSample}$ has excess $t$, and it
    intersects the $k$\th level, then all its points have depth at most
    $k + t(n/r)$. The expected number of vertices of
    $\ArrVD{\RSample}$ of depth at most $\alpha(t) = k + t(n/r)$ is
    \[
    \gamma(t) = O\pth{ \pth{ \frac{r\alpha(t)}{n}}^d
       \FZero\pth{\frac{n}{\alpha(t)}}}
    \]
    which also (asymptotically) bounds the number of cells in
    $\ArrVD{\RSample}$ having depth smaller than $\alpha(t)$. Let
    $X_t$ denote the number of cells with excess $t$ (or more) with
    depth at most $\alpha(t)$. Setting $c=O(d)$, we have by the
    polynomial decay lemma, that 
    \[
    \Ex{X_t} = O\pth{  \gamma(t)/t^{4d} } = O\pth{ 
       \pth{\frac{r\alpha(t)}{t^{4}n}}^d \FZero\pth{ \frac{n}{\alpha(t)}}}.
    \]
    Now, the number of cells of the cutting $\Cutting$ that have
    points with depth at most $k$ is bounded by
    \[
    Y = O \pth{ \sum_{t=0}^\infty X_t \cdot \pth{t \log t}^d }
    \]
    Thus, we have
    \begin{eqnarray*}
        \Ex{Y} &=& O \pth{ \sum_{t=0}^\infty \Ex{X_t} \cdot t^{O(d)} } %
        %
        =%
        O \pth{ \sum_{t=0}^\infty  \pth{\frac{r \alpha(t)}{t^c n}}^d
           \FZero\pth{\frac{n}{\alpha(t)}} \cdot t^{O(d)} }
        \\
        &=& %
        O \pth{ \pth{\frac{ r}{n}}^d
           \FZero\pth{\frac{n}{k}}
           \sum_{t=0}^\infty  t^{O(d)-c}
           \pth{{k + t (n/r)\VBig{0.5cm}}}^d
        } 
        =%
        O \pth{  \pth{\frac{ r}{n}}^d
           \FZero\pth{\frac{n}{k}}
           \pth{k + n/r}^d
           \sum_{t=0}^\infty  t^{O(d)-c}
        } \\
        &=&%
        O \pth{  \pth{\frac{ kr}{n} + 1}^d
           \FZero\pth{\frac{n}{k}}       
        },
    \end{eqnarray*}
    by setting $c$ to be sufficiently large.
\end{proof}

\bigskip

The above proves \thmref{shallow:cutting} by using replication to
represent weights.

\end{document}